\documentclass[12pt,twoside]{article}

\usepackage{float}
\usepackage{graphicx}
\usepackage{epstopdf}
\usepackage{graphicx}
\usepackage{epic}
\usepackage{multirow}
\usepackage{pst-poly}  
\usepackage{pst-plot}  
\usepackage{pst-poly}  
\usepackage{tikz}
\usepackage{xcolor}
\usetikzlibrary{arrows,shapes,chains}

\renewcommand{\paragraph}{\roman{paragraph}}
\usepackage[a4paper]{geometry}
\makeatletter
\renewcommand\title[1]{\gdef\@title{\reset@font\Large\bfseries #1}}
\renewcommand\section{\@startsection {section}{1}{\z@}%
	{-3.5ex \@plus -1ex \@minus -.2ex}%
	{2.3ex \@plus.2ex}%
	{\normalfont\large\bfseries}}
\renewcommand\subsection{\@startsection{subsection}{2}{\z@}%
	{-3ex\@plus -1ex \@minus -.2ex}%
	{1.5ex \@plus .2ex}%
	{\normalfont\normalsize\bfseries}}
\renewcommand\subsubsection{\@startsection{subsubsection}{3}{\z@}%
	{-2.5ex\@plus -1ex \@minus -.2ex}%
	{1.5ex \@plus .2ex}%
	{\normalfont\normalsize\bfseries}}

\def\@runningauthor{}\newcommand{\runningauthor}[1]{\def\runningauthor{#1}}
\def\@runningtitle{}\newcommand{\runningtitle}[1]{\def\runningtitle{#1}}

\renewcommand{\ps@plain}{%
	\renewcommand{\@evenhead}{\footnotesize\scshape \hfill\runningauthor\hfill}
	\renewcommand{\@oddhead}{\footnotesize\scshape \hfill\runningtitle\hfill}}
\newcommand{\Z}{\mathbb{Z}}
\newcommand{\F}{\mathbb{F}}
\newcommand{\x}{\mathbf{x}}

\newcommand {\D}{\mathbf{D}}

\newcommand {\C}{{\mathcal{C}}}
\newcommand {\ccc}{{\mathbf{c}}}

\newcommand {\0}{\mathbf{0}}

\g@addto@macro\bfseries{\boldmath}

\makeatother

\usepackage{amsthm,amsmath,amssymb}
\usepackage{cite}
\usepackage{graphicx}

\usepackage[colorlinks=true,citecolor=black,linkcolor=black,urlcolor=blue]{hyperref}

\theoremstyle{plain}
\newtheorem{theorem}{Theorem}
\newtheorem{lemma}[theorem]{Lemma}
\newtheorem{corollary}[theorem]{Corollary}
\newtheorem{proposition}[theorem]{Proposition}

\theoremstyle{definition}
\newtheorem{definition}[theorem]{Definition}
\newtheorem{example}[theorem]{Example}

\newtheorem{problem}[theorem]{Problem}

\theoremstyle{remark}
\newtheorem{remark}[theorem]{Remark}

\title{Some bounds on the cardinality of the $b$-symbol weight spectrum of codes
}

\runningtitle{Some bounds on the cardinality of the $b$-symbol weight spectrum of codes}

\author{Hongwei Zhu\thanks{Hongwei Zhu is with the Tsinghua Shenzhen International Graduate School, Tsinghua University, Shenzhen 518055, China. E-mail: zhwgood66@163.com, hongweizhu@sz.tsinghua.edu.cn.}
	\and Shitao Li\thanks{Shitao Li is with the School of Mathematical Sciences, Anhui University, Hefei 230601, China. E-mail: lishitao0216@163.com.}
	\and Minjia Shi\thanks{Minjia Shi is with the Key Laboratory of Intelligent Computing and Signal Processing, Ministry of Education, School of Mathematical Sciences, Anhui
		University, Hefei 230601, China, and also with the State Key Laboratory of Information Security, Institute of Information Engineering, Chinese Academy
		of Sciences, Beijing 100093, China. E-mail: smjwcl.good@163.com.}
	\and Shu-Tao Xia\thanks{Shu-Tao Xia is with the Tsinghua Shenzhen International Graduate School, Tsinghua University, Shenzhen 518055, China. E-mail: xiast@sz.tsinghua.edu.cn.}
	\and Patrick Sol\'e\thanks{Patrick Sol\'e is with the Centrale Marseille, CNRS, I2M, Aix Marseille University, Marseille 13009, France, E-mail: sole@enst.fr.}
}

\runningauthor{}

\date{}

\begin{document}

	\maketitle
	
	\thispagestyle{empty}
	
	\begin{abstract}
The size of the Hamming distance spectrum of a code has received great attention in recent research. The main objective of this paper is to extend these significant theories to the $b$-symbol distance spectrum. We examine this question for various types of codes, including unrestricted codes, additive codes, linear codes, and cyclic codes, successively.
For the first three cases, we determine the maximum size of the $b$-symbol distance spectra of these codes smoothly. For the case of cyclic codes, we introduce three approaches to characterize the upper bound for the cardinality of the $b$-symbol weight spectrum of cyclic codes, namely the period distribution approach, the primitive idempotent approach, and the $b$-symbol weight formula approach. As two by-products of this paper, the maximum number of symplectic weights of linear codes is determined, and a basic inequality among the parameters $[n,k,d_H(\C)]_q$ of cyclic codes is provided.
	\end{abstract}

	\section{Introduction}	
	\hspace{0.5cm}One of the most important topics in coding theory is to characterize the parameters of codes with a given structure. The parameters of a code include its length, size, minimum distance, cover radius, and so on. In 1973, Delsarte \cite{DE} introduced four fundamental parameters of an unrestricted code $C$: minimum distance $d(C)$, the number $s(C)$ of distinct distances of $C$, $d(C')$, and $s(C')$, where $C'$ is the MacWilliams transform of $C$. In the linear case, these four fundamental parameters are reduced to the minimum weight and the weight spectrum of the given code and of its dual code. Delsarte \cite{DE} pointed out the combinatorial significance of these four fundamental parameters and showed that they have close connection with orthogonal arrays and combinatorial designs. He also showed that $s(C)$ is intimately related to other parameters of the given $C$. For instance, it can estimate the size $M$ of a code $C$ as follows:
	\begin{equation}\label{de}
	q^n\Bigg/\sum_{j=0}^{s'(C)}{n \choose j}(q-1)^j \leq M\leq \sum_{j=0}^{s(C)}{n \choose j}(q-1)^j,
	\end{equation}
where $n$ denotes the length of $C$ and $q$ denotes the size of the alphabet. It was also shown in the same paper \cite{DE} that the covering radius of $C$ is less than or equal to $s(C')$.

Based on these properties, it is always an interesting topic to study the four fundamental parameters of a code. Since the beginning of coding theory, most scholars tended to focus on the minimum distance of a code, but often ignored the other two parameters due to their difficulty. Recently, there has been an increasing interest in these two parameters of codes. Ezerman {\em et al.} \cite{MDS-IT-2011} determined the weight spectrum of linear maximum distance separable (MDS) codes with parameters given by the MDS conjecture. Alderson \cite{Ald1} developed this result to general MDS codes including non-linear MDS codes.
However, it is very hard to determine the exact value of $s(C)$ for a general code $C$. A simpler objective is to give some bounds on the exact value of $s(C)$.
For this topic, Shi {\it et al.} \cite{SZS} showed the maximum value of $s(C)$ is $\frac{M(M-1)}{2}$ when $C$ is an unrestricted code with size $M$. For a linear $[n,k]_q$ code $C$, the authors in the same paper \cite{SZS} provided an upper bound on $s(C)$, namely $s(C)\leq \frac{q^k-1}{q-1}$, and showed that this bound is achievable for $q=2$ or $k=2$. They also conjectured that this bound is achievable for general linear codes. Subsequently, this conjecture was proven in \cite{AN} by two different methods. The work of \cite{AN} and \cite{SZS} also raises a derivative problem, which is the existence of maximum weight spectrum codes \cite{Ald2}.

If the linear code $C$ is given some interesting structures, how will the size of the weight spectrum of $C$ change? For this problem, Shi {\em et al.} \cite{SLNS,SNS} obtained some lower and upper bounds on the size of the weight spectrum of cyclic codes and quasi-cyclic codes. Chen and Zhang \cite{ChenZhang} proposed a tight upper bound on the size of the weight spectrum of cyclic codes by utilizing the automorphism group of cyclic codes. Chen {\it et al.} \cite{ChenFu} further refined the results of \cite{ChenZhang} by using a larger automorphism group of cyclic codes than the group used in \cite{ChenZhang}. Very recently, the methods in \cite{ChenFu,ChenZhang} respectively extended to quasi-cyclic codes \cite{LSL-arxiv} and constacyclic codes \cite{ZC-FFA-1}. Luo {\em et al.} \cite{LCFL-IT} established two sharp upper bounds on the size of the weight spectrum of a linear code with prescribed automorphism. Note that the aforementioned works are all conducted under the Hamming metric.
Chen {\it et al.} in \cite{ChenFu} mentioned the following open problem:
\begin{problem}
It would be interesting to find tight upper bounds on the number of symbol-pair weights of cyclic codes.
\end{problem}
The main objective of this paper is to solve aforementioned problem, but we are no longer limited to symbol-pair weight and cyclic codes. Instead, we replace them with $b$-symbol weight and general codes respectively. It is not difficult to imagine that this will be more difficult but also more interesting.

 The $b$-symbol metric is a generalization of the Hamming metric, introduced by Cassuto and Blaum \cite{CB1,CB} and motivated by certain types of computer memories. Despite being proposed only a decade ago, extensive research has been conducted on the $b$-symbol metric. For instance, bounds on parameters of codes \cite{CB1,CB,CL,CLL,Eli}, algebraic constructions \cite{CL}, asymptotic performance \cite{CL}, the $b$-symbol weight distributions of cyclic codes \cite{SOS,Z,ZHW}, and decoding \cite{Yaa,Yaa1} have all been explored under the $b$-symbol metric. Notably, recent works on $b$-symbol MDS codes \cite{C+,DGZ,DTG,KZL} have demonstrated the advantages of the $b$-symbol metric over the Hamming metric, especially in encoding information with low redundancy. Additionally, Shi {\it et al.} \cite{BUG} investigated the connections between the $b$-symbol metric and the $r$-th generalized Hamming metric. Liu and Pan studied the generalized $b$-symbol weight in \cite{LiuPan1, LiuPan2}, which can be viewed as a fusion of the $r$-th generalized Hamming weight and $b$-symbol weight.

In this paper, we explore the size of the $b$-symbol weight spectra of unrestricted codes, additive codes, linear codes, and cyclic codes in sequence. For the first three cases, we determine the maximum size of the $b$-symbol distance spectra of these codes smoothly. When the research object is cyclic codes, some upper bounds on the $b$-symbol weight spectrum of cyclic codes are provided by using nice algebraic structures of cyclic codes and unifying the two approaches presented in \cite{SLNS,ChenZhang}. We also demonstrate that the approaches \cite{SLNS,ChenZhang} are equally effective for studying the $b$-symbol weight spectrum of cyclic codes. Additionally, we propose a third approach for investigating the $b$-symbol weight spectrum of cyclic codes, which involves analyzing the $b$-symbol weight of cyclic codewords. Through a performance analysis of the three approaches, we have found that the upper bounds provided by the first and second approaches are equivalent, while those given by the third approach are more accurate but have limited applicability. As two by-products of this paper, we answer how many symplectic weights can a linear code have at most (Corollary \ref{addbound}), and provide a fundamental lower bound on the minimum Hamming distance of cyclic codes (Theorem \ref{the9}). To the best of our knowledge, this lower bound has not been previously established. Although this lower bound is straightforward to derive, it holds significant importance. Previous studies on the minimum Hamming distance of cyclic codes have primarily relied on their defining sets, whereas the lower bound given by Theorem \ref{the9} is independent of defining sets.

The remainder of the paper is organized as follows: In Section II, we introduce the research strategy of this paper along with some notation and definitions. Section III delves into the study of the $b$-symbol weight spectrum of unrestricted codes. Section IV addresses the case of additive codes. We also examine the size of the symplectic weight spectrum of linear codes in Section IV, as the symplectic metric is closely related to additive codes. Section V focuses on exploring the $b$-symbol weight spectrum of linear codes. In Section VI, we further delve into the $b$-symbol weight spectrum of cyclic codes. Finally, Section VII concludes the paper. We have included some technical but significant results in the appendix to illustrate the advantages of the period distribution approach.

	\section{Preliminaries}

Our research strategy unfolds as follows: We initially explore codes with the weakest constraints, followed by codes with more constraints. In other words, we present the bounds on the cardinality of the $b$-symbol weight spectrum in the sequence of unrestricted codes, additive codes, linear codes, and cyclic codes. Their hierarchical relationship is depicted in Figure 1. Among these, we are particularly interested in the $b$-symbol weight spectrum of cyclic codes. It's important to note that the cyclic codes discussed in this paper are all linear codes.
	\begin{figure}[htbp]
		\centering
		\includegraphics[width=0.7\textwidth]{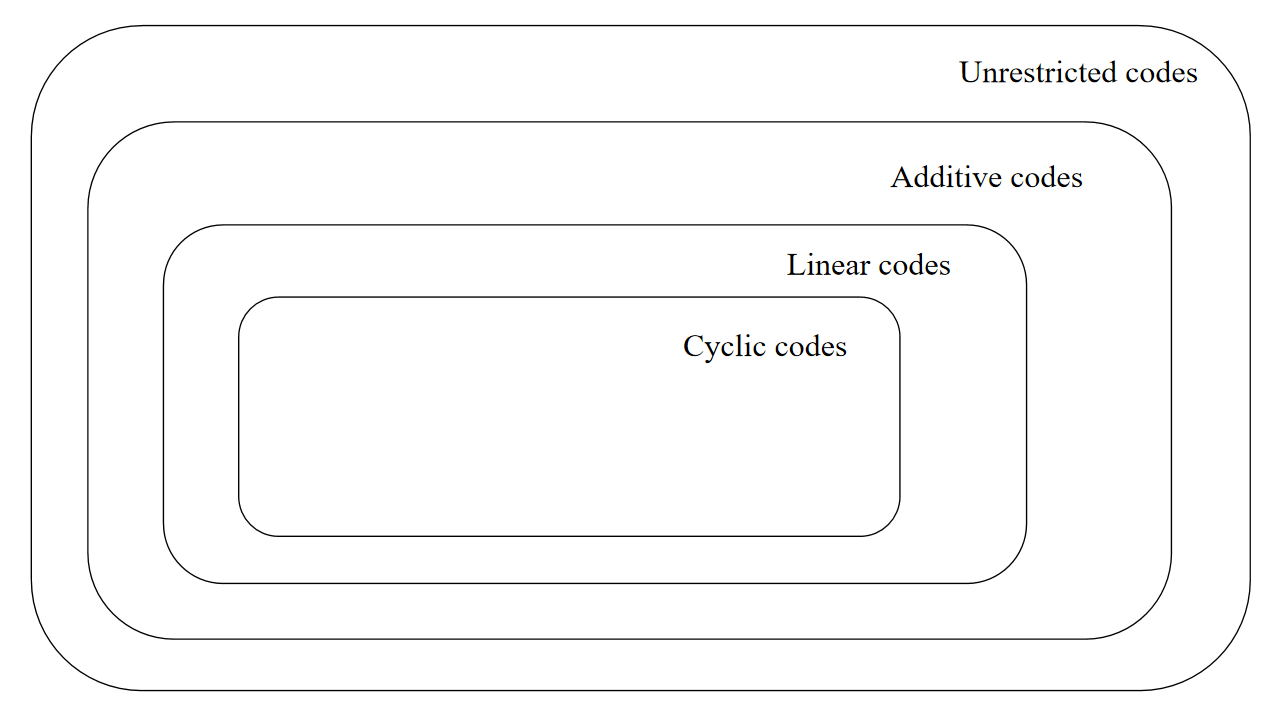}
		\caption{The Venn diagram of these codes}
		\label{fig1}
	\end{figure}
	\subsection{Some notation fixed throughout this paper}
	\hspace{0.5cm}We begin by introducing some notation and definitions:
	\begin{itemize}
  \setlength{\itemsep}{3pt}
\setlength{\parsep}{0pt}
\setlength{\parskip}{0pt}
		\item For any finite set $A$, $|A|$ denotes the cardinality of $A$, and $A^{*}$ represents $A$ without the element $0$.
		\item Let $m$ and $s$ be positive integers, and $p$ be a prime. We define $q:=p^s$ and $Q:=q^m$.
		\item $\mathbb{Z}_n=\{0,1,2,\ldots,n-1\}$ denotes the ring of integers modulo $n$.
		\item$\F_q$ denotes the finite field with $q$ elements, and $\F_q^n$ represents the vector space of all $n$-tuples over the finite field $\F_q$.
		\item $[i]$ denotes the set $\{1,2,\ldots,i\}$.
		\item $C$, $\mathfrak{C}$, and $\C$ denote the unrestricted code, the additive code, and the linear code, respectively.
		\item $\alpha$ denotes a primitive element of $\F_{Q}$.
		\item $M_a(x)$ denotes the minimal polynomial of $a\in \F_{Q}$ over $\F_q$.
		\item ${\rm{Tr}}_{Q/q}$ denotes the relative trace function from $\F_{Q}$ to $\F_q$.
		\item $C_i$ denotes the $q$-cyclotomic coset modulo $n$ containing $i$.
		\item $\Gamma$ is the set of all coset leaders of the $q$-cyclotomic cosets modulo $n$.
		\item Let  $\mathbf{\overline{a}}$ denote a sequence.
		\item Let ${\rm Aut}(\C)$ denote the automorphism group of $\C$.
		\item $\mu(\cdot)$ denotes the M$\ddot{o}$bius function.
		\item $\phi(\cdot)$ denotes Euler's totient function.
	\end{itemize}

	\subsection{The Hamming metric and the $b$-symbol metric}
	\hspace{0.5cm}The research objects in this paper fall into four categories:
	\begin{itemize}
  \setlength{\itemsep}{3pt}
\setlength{\parsep}{0pt}
\setlength{\parskip}{0pt}
		\item Unrestricted code with parameters $(n,M)_q$: An unrestricted code $C$ of length $n$ over $\F_q$ is a subset of $\F_q^n$, where $M=|C|$.
		\item Additive code with parameters $(n,q^{k_0})_{q^t}$: An additive code $\mathfrak{C}$ of length $n$ over $\F_{q^t}$ is a subset of $\F_{q^t}^n$ closed under $\F_q$-linear vector addition, where $q^{k_0}=|\mathfrak{C}|$ and $0\leq k_0\leq tn$. So $\mathfrak{C}$ is an $\F_q$-linear subgroup of $\F_{q^t}^n$.
		\item Linear code with parameters $[n,k]_q$: A linear code $\C$ of length $n$ over $\F_q$ is a $k$-dimensional vector subspace of $\F_q^n$.
		\item Cyclic code with parameters $[n,k]_q$: Cyclic codes are a very important subclass of linear codes, and their specific definition and algebraic properties will be introduced in detail in Section \ref{cycliccodes}.
	\end{itemize}

	
		The Hamming weight $w_H(\mathbf{x})$ of a vector $\mathbf{x}$ is defined as the number of non-zero coordinates in $\mathbf{x}$. The Hamming distance $d_H(\mathbf{x},\mathbf{y})$ between two vectors $\mathbf{x}$ and $\mathbf{y}$ is defined as the number of coordinates in which $\mathbf{x}$ and $\mathbf{y}$ differ.
	For a code $C$, we define the  Hamming distance spectrum and the Hamming weight spectrum of $C$ as $\mathbf{D}_H(C)=\{d_H(\ccc,\ccc^{\prime})|\ccc,\ccc^{\prime}\in C{\rm ~and~}\ccc\neq\ccc^{\prime}
	\}$ and $\mathbf{W}_H(C) = \{w_H(\ccc)|\ccc\in C\setminus\{\0\}\},$ respectively.

	The $b$-symbol metric is a generalization of the Hamming metric. Its definition is as follows:
		Let $b$ be a positive integer with $1\leq b\leq n.$ The $b$-symbol weight $w_b(\x)$ is the Hamming weight of $\pi_b(\x)$, where $\pi_b(\x)\in (\F_q^b)^n$ and
		$$\pi_b(\x)=((x_0,\ldots,x_{b-1}),(x_{1},\ldots,x_{b}),
		\cdots,(x_{n-1},\ldots,x_{b+n-2({\rm mod}~n)})).$$
		The $b$-symbol distance $d_b(\x,\mathbf{y})$ is the Hamming weight of $\pi_b(\x-\mathbf{y})$. When $b=1$, $w_1(\x)=w_H(\x)$ and $d_1(\x,\mathbf{y})=d_H(\x,\mathbf{y})$.
	Similarly, we define the $b$-symbol distance spectrum and the $b$-symbol weight spectrum of a code $C$ as $\mathbf{D}_b(C) = \{d_b(\ccc,\ccc^{\prime})|\ccc,\ccc^{\prime}\in C{\rm ~and~}\ccc\neq\ccc^{\prime}\}$ and $\mathbf{W}_b(C) =
	\{w_b(\ccc)|\ccc\in C\setminus\{\0\}\},$ respectively.

	\subsection{Group action and automorphism group}

Assume that a finite group acts on a finite set $A$. For each $a\in A$, $Ga=\{ga|g\in G\}$ is called an orbit of this group action containing $a$ ($G$-orbit for short). The set of all orbits of $G$ on $A$ is denoted as $G\backslash A=\{Ga|a\in A\}$.  A monomial matrix is a square matrix with exactly one nonzero entry in each row and column. A monomial matrix $M$ can be written either in the form $DP$ or $PD$, where $P$ is a permutation matrix and $D$ is an invertible diagonal matrix.
The automorphism group ${\rm Aut}(\C)$ of a $q$-ary linear code $\C$ of length $n$ is the set consisting of all $n\times n$ monomial matrices $A$ over $\F_q$ such that $\C A=\C$.

	Let $\mathsf{G}$ be a subgroup of ${\rm Aut}(\C)$. It can be observed that the number of orbits belonging to ${\rm Aut}(\C)$ on $\C^*$ is not greater than the number of orbits belonging to $\mathsf{G}$ on $\C^*$.
	
	Set $\tau$ and $\sigma_a$ be two $\F_q$-linear maps on $\frac{\F_q[x]}{(x^n-1)}$, respectively:
	\begin{eqnarray*}
		\tau: \frac{\F_q[x]}{(x^n-1)}&\longrightarrow&\frac{\F_q[x]}{(x^n-1)}\\
		\sum_{i=0}^{n-1}a_ix^i&\longmapsto&\tau\left(\sum_{i=0}^{n-1}a_ix^i\right)=\sum_{i=0}^{n-1}a_ix^{i+1}({\rm mod}~x^n-1),\\
		\sigma_a: \frac{\F_q[x]}{(x^n-1)}&\longrightarrow&\frac{\F_q[x]}{(x^n-1)}\\
		\sum_{i=0}^{n-1}a_ix^i&\longmapsto&\sigma_a\left(\sum_{i=0}^{n-1}a_ix^i\right)=\sum_{i=0}^{n-1}aa_ix^{i}({\rm mod}~x^n-1).
	\end{eqnarray*}
	It is not difficult to check that for any cyclic code $\C$, the group $\left\langle \tau, \{\sigma_a|a\in\F_q^*\}\right\rangle$ is a subgroup of ${\rm Aut}(\C)$.

\subsection{Some basics about  periodic sequences}
	\hspace{0.5cm}We introduce the definition of runs in sequences in the following.
	For $\lambda\in \F_q$ and $\eta, \xi\in \F_q\setminus\{\lambda\}$, if $\eta,\underbrace{\lambda,\ldots,\lambda}_m, \xi$ appears in the sequence $\mathbf{\overline{a}}$, then we say that $\underbrace{\lambda,\ldots,\lambda}_m$ is a run of $\lambda$'s of length $m$.
	In particular, we set $\mathbf{0}_i=(\alpha,\underbrace{0,\ldots,0}_i,\beta)$, where $\alpha, \beta\in\F_q^*$ since it plays a key role in computing the $b$-symbol weight of a vector later on.
	We generalize the definition of runs in sequences to the definition of runs in vectors. For any vector $\mathbf{a}=(a_0,a_1,\ldots,a_{n-1})\in \F_q^n$, we define a circumferential vector $cir(\mathbf{a})$ as follows:
	\vspace{0.5cm}
	\begin{center}\begin{tikzpicture}
			$\def \radius {1.8cm}
			\def\startDegree{50}
			\def\n{10}
			\foreach \s in {1,...,\n}
			{
				\draw[<-, >=latex] ({360/\n * (\s - 1)-\startDegree}:\radius)
				arc ({360/\n * (\s - 1)-\startDegree}:{360/\n * (\s - 1)-\startDegree+360/\n}:\radius);
			}
			$,
			\put(-12,56){\fontsize{12}{2.5}\selectfont $ a_0$}
			\put(25,50){\fontsize{12}{2.5}\selectfont $ a_1 $}
			\put(50,22){\fontsize{12}{2.5}\selectfont $ a_2 $}
			\put(-6,-42){\fontsize{12}{2.5}\selectfont $ \ldots $}
			\put(-80,12){\fontsize{12}{2.5}\selectfont $ a_{n-2} $}
			\put(-60,42){\fontsize{12}{2.5}\selectfont $ a_{n-1} $}
		\end{tikzpicture}
	\end{center}
	In fact, $cir(\mathbf{a})$ can be regarded as a periodic sequence with a period of $n$, so we can call $(a,\cdots,a)$ in $cir(\mathbf{a})$ that appears $m$ consecutive times as a run of $a$ in the vector with length $m$. Let
	$\Psi(\mathbf{a},\mathbf{0}_i)$ denote the number of occurrences of $\mathbf{0}_i$ on the circumferential vector $cir(\mathbf{a})$.
	\begin{definition}
		The $\mathbf{0}$'s run distribution of $\mathbf{a}$ is defined by
		$\{\Psi(\mathbf{a},\mathbf{0}_1), \Psi(\mathbf{a},\mathbf{0}_{2}),\ldots, \Psi(\mathbf{a},\mathbf{0}_{n})\}.$
	\end{definition}

	For any vector $\mathbf{a}=(a_0,a_1,\ldots,a_{n-1})\in\F_q^n,$ by the definition of $b$-symbol weight, then 
$w_b(\mathbf{a})=n-\left|\{i|a_i=a_{i+1}=\cdots=a_{i+b-1}=0,0\leq i\leq n-1\}\right|.$
	If the $\mathbf{0}$'s run distribution of $\mathbf{a}$ is given, we have the following formula to calculate the $b$-symbol weight of $\mathbf{a}$.
	
	\begin{lemma}\cite{BUG}\label{wbformula11}
		For any vector $\mathbf{a}=(a_0,a_1,\ldots,a_{n-1})\in\F_q^n,$ we have
		\begin{equation}\label{formula}
			w_b(\mathbf{a})=n-\sum\limits_{i=b}^{n-1}(i-b+1)\cdot \Psi(\mathbf{a},\mathbf{0}_i).
		\end{equation}
	\end{lemma}
	
	\begin{example}
		Let $\mathbf{a}=(01001000100)$. Then the $\mathbf{0}$'s run distribution of $\mathbf{a}$ is
		$\{\Psi(\mathbf{a},\mathbf{0}_2)=1,
		\Psi(\mathbf{a},\mathbf{0}_3)=2,
		\Psi(\mathbf{a},\mathbf{0}_i)=0, i\neq 2,3\}.$ The $b$-symbol weight of $\mathbf{a}$ is
		\begin{equation*}
			w_b(\mathbf{a})=\left\{
			\begin{array}{ll}
				11-(2\time 1+3\time 2)=3, & \hbox{$b=1$;} \\
				11-(1\time 1+2\time 2)=6, & \hbox{$b=2$;} \\
				11-(1\time 2)=9, & \hbox{$b=3$;}\\
				11, & \hbox{$b=4$.}
			\end{array}
			\right.
		\end{equation*}
	\end{example}
	
	Assume that $\ccc$ is a codeword of $C$.
	Let $\vartheta(\ccc)$ be the maximum $\mathbf{0}$'s run length of $\ccc$ and $\theta=\max\{\vartheta(\ccc)|\ccc\in C\backslash\{\mathbf{0}\}\}$. The parameter $\theta$ is called the maximum $\mathbf{0}$'s run length of $C$. A partition of $C\backslash\{\mathbf{0}\}$ associated with $\theta$ can be expressed as follows:
$C\backslash\{\mathbf{0}\}=\bigsqcup_{i=0}^{\theta}C(i),$
	where $C(i)=\{\ccc|\ccc\in C\backslash\{\mathbf{0}\} {\rm ~and~} \vartheta(\ccc)=i\}$ and $\bigsqcup$ is a disjoint union.

	\section{Unrestricted codes}
	\hspace{0.5cm}In this section, our focus is on determining the maximum number of $b$-symbol distances that an unrestricted code can have. Let $\Sigma$ be a finite alphabet. We will begin by establishing a fundamental bound for $|\mathbf{D}_b(C)|$.
	\begin{proposition}
		Let $C$ be an unrestricted code with parameters $(n,M)$ over $\Sigma$. Then
		\begin{equation}\label{bde}
			M\leq\sum_{j=0}^{|\mathbf{D}_b(C)|}{n\choose j}(|\Sigma|^b-1)^j.
		\end{equation}
	\end{proposition}
\begin{proof}
	 Following the approach of the proof in \cite{DE} for Inequality (\ref{de}), we can readily derive the desired result. It's important to highlight that at this juncture, the distance between any two codewords of $C$ is considered to be the $b$-symbol distance rather than the Hamming distance.
\end{proof}
	Let $G_1$ be the set consisting of all unrestricted codes with size $M$ over $\Sigma$. We define $N(M,b)$ and $N(n,M,b)$ as follows:
	\begin{itemize}
  \setlength{\itemsep}{3pt}
\setlength{\parsep}{0pt}
\setlength{\parskip}{0pt}
		\item $N(M,b):=\max\{|\mathbf{D}_b(C)||C\in G_1\}.$
		\item $N(n,M,b):=\max\{|\mathbf{D}_b(C)||C\in G_1~{\rm of~length}~n\}.$
	\end{itemize}
	Note that we focus on $N(M,b)$ and $N(n,M,b)$ instead of $N(M,|\Sigma|,b)$ and $N(n,M,|\Sigma|,b)$ because the alphabet size $\Sigma$ does not affect these combinatorial functions. 
	\begin{theorem}
		For all integers $M\geq2$, we have
		$$N(M,b)={M\choose 2}.$$
	\end{theorem}
	\begin{proof}
We just need to construct an unrestricted code $C$ with $|\D_b(C)|={M \choose 2}$ because $|\D_b(C)|\leq{M \choose 2}$ is obvious.
According to \cite[Theorem 6]{SZS}, there exists an unrestricted code $C$ with $|\D_{1}(C)|={M\choose 2}.$ Let $C^{\prime}=\{f(\ccc)|\ccc=(c_0,c_1,\ldots,c_{n-1})\in C\},$ where
		$$f(\ccc)=(c_0,\underbrace{0,\ldots,0}_{b-1},
		c_1,\underbrace{0,\ldots,0}_{b-1},\ldots,
		c_{n-1},\underbrace{0,\ldots,0}_{b-1}).$$
Assume that $\D_1(C)=\{d_{i,j}|1\leq i< j\leq M\}$.  Then we have
		$\D_b(C^{\prime})=\{b\cdot d_{i,j}|1\leq i<j\leq M\}$. The size of $\D_b(C^{\prime})$ equals ${M \choose 2}$ since $b\cdot d_{i_1,j_1}\neq b\cdot d_{i_2,j_2}$ with $(i_1,j_1)\neq (i_2,j_2).$ Therefore, we obtain the desired result.
	\end{proof}
	
	If there are no constraints on the length $n$, we can append consecutive zeros to ensure that the $b$-symbol distance between any two codewords in $C$ is distinct. To facilitate the proof of the subsequent results, we introduce the definition of difference sets with Singer parameters and Golomb ruler below \cite{BJ,D}.
	
	\begin{definition}
		A difference set in an abelian group $A$ is said to have Singer parameters if their parameters are of the form $\left(\frac{l^m-1}{l-1},\frac{l^{m-1}-1}{l-1},\frac{l^{m-2}-1}{l-1}\right)$ or $\left(\frac{l^{m}-1}{l-1},l^{m-1},l^{m-2}(l-1)\right)$,
		where $m$ and $l$ are positive integers.
	\end{definition}
	The first class of difference sets with Singer parameters were discovered by Singer in 1938, and are described in the following theorem.

	\begin{theorem}\cite{Singer1938}\label{Singer1938}
		Let $m\geq 3$ be a positive integer. Let $\alpha$ be a generator of $\F_{Q}^*$. Put $n=\frac{Q-1}{q-1}$. Define $D=\left\{0\leq i\leq n|{\rm Tr}_{Q/q}(\alpha^i)=0\right\}\subset \Z_n.$ Then $D$ is a difference set in $(\Z_n,+)$ with parameters $\left(\frac{Q-1}{q-1},\frac{q^{m-1}-1}{q-1},\frac{q^{m-2}-1}{q-1}\right)$.
	\end{theorem}
	With the help of Theorem \ref{Singer1938}, we can obtain the following result.
	\begin{theorem}\label{the20}
		For all positive integers $M$, $b$ and all nonnegative integers $m$, we have
		$N(M,b)=N(n_0(M,b)+m,M,b)$
		where $n_0(M,b)$ is a function about $M$ and $b$.
		Moreover, $n_0(M,b)\leq 2b{M\choose 2}+b$ if $M-1$ is a prime power.
	\end{theorem}
	\begin{proof}
		The first claim is obvious. Assume that $M=q+1$, where $q$ is a prime power. By Theorem \ref{Singer1938}, there exists a Singer difference set $S=\{a_1,a_2,\ldots,a_{q+1}\}$ with parameters $(q^2+q+1,q+1,1)$.
		Let $C=\{\ccc_i|1\leq i\leq q+1\}$, where
		$$\ccc_i=(\underbrace{1\underbrace{00\cdots0}_{b-1}~1
			\underbrace{00\cdots0}_{b-1}
			~\cdots1\underbrace{00\cdots0}_{b-1}}_{a_i}~
		\underbrace{00\cdots0}_{b\cdot(q^2+q+1)-b\cdot a_i}) \hbox{~with $1\leq i\leq {q+1}$}.$$
		According to the definition of $b$-symbol distance, we obtain the $b$-symbol distance spectrum
		$\D_b(C)=\{b\cdot|{a_i-a_j}||1\leq i\neq j\leq q+1\}$ and $|\D_b(C)|={q+1\choose 2}.$ Therefore, $b(q^2+q+1)=2b{q+1\choose 2}+b$ is an upper bound for $n_0(q+1,b).$
	\end{proof}
	Indeed, $n_0(M,b)$ represents the minimum length $n$ for which $N(n,M,b)=N(M,b)$.
The value of $N(M,b)$ exhibits an intriguing connection with another significant combinatorial concept, known as the Golomb ruler.
	\begin{definition}
		A Golomb ruler $GR(M,v)$ is a set of $M$ integers
		$$\{a_1,a_2,\ldots,a_{M}\}\subseteq\{0,1,\ldots,v\}$$
		whose pairwise differences in absolute value take ${M\choose 2}$ distinct values exactly once. By convention, we assume that the sequence $i\mapsto a_i$ is increasing and $a_1=0$, $a_M=v.$
	\end{definition}

	The function $G(M)$, denoting the length of an optimal Golomb ruler, represents the smallest possible length of a ruler with $M$ marks.
As a consequence, we can establish an upper bound for $n_0(M,b)$.
	\begin{proposition}\label{pro21}
		Let notations be as above. Then $n_0(M,b)\leq b\cdot G(M).$
	\end{proposition}
	\begin{proof}
		The inequality $n_0(M,1) \leq G(M)$ follows directly from the definition of $G(M)$. By mimicking the proof of Theorem \ref{the20}, we obtain the desired result.
	\end{proof}

	\section{Addtive codes}
	\hspace{0.5cm}In this section, we investigate the number of $b$-symbol distances that an additive code with parameters $(n,q^{k_0}){q^t}$ can have. Let $G_2$ denote the set comprising all additive codes with size $|\mathfrak{C}|=q^{k_0}$ over $\mathbb{F}{q^t}$. We define $A(q^{k_0},q^t,b)$ as follows:
	$$A(q^{k_0},q^t,b):=\max\left\{|\mathbf{D}_b(\mathfrak{C})||\mathfrak{C}\in G_2\right\}.$$
	
	\begin{theorem}\label{additivecase}
		Let $\mathfrak{C}$ be an additive code with parameters $(n,q^{k_0})_{q^t}$, then we have
		$A(q^{k_0},q^t,b)\\=\frac{q^{k_0}-1}{q-1}.$
	\end{theorem}
	\begin{proof}
		Since $\mathfrak{C}$ is an $\F_q$-linear code, by mimicking the construction of \cite{AN}, we obtain
		$A(q^{k_0},q^t,1)=\frac{q^{k_0}-1}{q-1}.$ For $a\in \F_{q}^*$ and $\mathbf{c}=(c_0,c_1,\ldots,c_{n-1})\in \mathfrak{C}$, one has $w_b(\mathbf{c})=w_b(a\cdot{\mathbf{c}})$ since $\mathbf{c}$ and $a\cdot\mathbf{c}$ have the same $\0$'s run distribution. Then $$A(q^{k_0},q^t,b)\leq \frac{q^{k_0}-1}{q-1}=A(q^{k_0},q^t,1).$$ Let $\mathfrak{C}$ be an additive code such that $\mathfrak{C}$ has $\frac{q^{k_0}-1}{q-1}$ nonzero Hamming weights. For any $\ccc=(c_0,c_1,\ldots, c_{n-1})\in \mathfrak{C}$, define
		$$f(\mathbf{c})=(c_0,\underbrace{00\cdots0}_{b-1},c_1,
		\underbrace{00\cdots0}_{b-1},
		\ldots,\underbrace{00\cdots0}_{b-1},c_{n-1},
		\underbrace{00\cdots0}_{b-1}).$$
		Let $\mathfrak{C}^{\prime}=\{f(\mathbf{c})|\mathbf{c}\in \mathfrak{C}\}.$ It is easy to verify that $\mathfrak{C}^{\prime}$ is an additive code over $\F_q$ with $\frac{q^k-1}{q-1}$ nonzero $b$-symbol weights. This implies that $A(q^{k_0},q^t,b)\geq \frac{q^{k_0}-1}{q-1}.$ The result holds.
	\end{proof}
	Given the numerous similarities between additive codes and linear codes, we will focus our detailed study on linear codes. Indeed, examining the $b$-symbol weight of a linear code is equivalent to analyzing the Hamming weight of an additive code, for the following reasons.
	Set $\mathbf{a}:=(a_1,a_2,\ldots,a_n)$.
	Let $\phi$ be a map from $(\F_{q}^b)^n$ to $(\F_{q^b})^n$ as follows.
	\begin{eqnarray}\label{phi}
		\phi: (\F_{q}^b)^n&\longrightarrow&(\F_{q^b})^n \\
		\pi_b(\mathbf{a})&\longmapsto&\phi(\pi_b(\mathbf{a}))=\left(\sum_{i=1}^{b-1}\omega^{b-1}a_i,\sum_{i=1}^{b-1}\omega^{b-1}a_{i+1},\ldots,\sum_{i=1}^{b-1}\omega^{b-1}a_{b+i}\right),\nonumber
	\end{eqnarray}
	where the indices are taken modulo $n$, and $\omega$ is a primitive element of $\F_{q^b}$. Therefore, the $b$-symbol weight spectrum of a linear code $\C$ with parameters $[n,k]_q$ is equivalent to the Hamming weight spectrum of the additive code $\phi(\C)$ with parameters $(n,q^k)_{q^b}$, where $\phi(\C)$ is an $\F_{q}$-linear subgroup of $\F_{q^b}^n$.
As a by-product, Theorem \ref{additivecase} reveals the maximum number of symplectic weights that a linear code can have. Before that, we introduce the definition of symplectic metric as follows: For any vector $\x=(x_1,x_2,\ldots,x_{2n})$ and $\mathbf{y}=(y_1,y_2,\ldots,y_{2n})\in \F_q^{2n}$, the symplectic weight of $\x$ is defined by
		$w_{\rm sym}(\x)=|\{i|(x_i,x_{n+i})\neq(0,0),1\leq i\leq n\}|.$ The minimum symplectic distance of a linear code $\C$ is $d_{\rm sym}(\C)=\min\{w_{\rm sym}(\x-\mathbf{y})|\x,\mathbf{y}\in \C {\rm~ and~} \x \neq \mathbf{y}\}.$
	The symplectic metric has garnered significant attention from scholars due to its close relationship with the construction of quantum codes \cite{Cald,Ketk}.
	\begin{corollary}\label{addbound}
		Let $L^{\rm sym}(k,q)$ be the maximum number of nonzero symplectic weights a linear code of dimension $k$ over $\F_q$ may have. Then $L^{\rm sym}(k,q)=\frac{q^k-1}{q-1}.$
	\end{corollary}
	\begin{proof}
		Let  $\phi^{\prime}$ be a map from $\F_{q}^{2n}$ to $\F_{q^2}^n$ as follows.
		\begin{eqnarray}\label{phi}
			\phi^{\prime}: \F_{q}^{2n}&\longrightarrow& \F_{q^2}^n \\
			(x_1,x_2,\ldots,x_{2n})&\longmapsto&(x_1+w^{\prime}x_{n+1},\ldots,x_n+w^{\prime}x_{2n}),\nonumber
		\end{eqnarray}
		where  $\langle w^{\prime}\rangle=\F_{q^2}^*$. Therefore, the symplectic weight spectrum of a linear code with parameters $[2n,k]_q$ is the same as the Hamming weight spectrum of an additive code with parameters $(n,q^k)_{q^2}$. From Theorem \ref{additivecase}, we have
		$L^{\rm sym}(k,q)=A(q^{k},q^2,1)=\frac{q^k-1}{q-1}.$
		This completes the proof.
	\end{proof}

	\section{Linear codes}
	\hspace{0.5cm}The $b$-symbol distance spectrum $\mathbf{D}_b(\C)$ and the $b$-symbol weight spectrum $W_b(\C)$ coincide when $\C$ is linear. In this section, we investigate the inquiry regarding the maximum number of nonzero $b$-symbol weights a linear code can possess.	
Let $L(k,q,b)$ denote the maximum number of nonzero $b$-symbol weights that a code of dimension $k$ over $\F_q$ may possess.
		An $[n,k]_q$ code $\C$ is called a maximum $b$-symbol weight spectrum code (M$b$SW code) if $\C$ has $L(k,q,b)$ nonzero $b$-symbol weights.
	The subsequent result demonstrates that $L(k,q,b)$ remains independent of $b$ when there is no restriction on the length $n$ of $\C$.
	\begin{theorem}\label{thm8}
		For all prime powers $q$, and all integers $k\geq 1$, we have
		$L(k,q,b)=\frac{q^k-1}{q-1}.$
	\end{theorem}
	\begin{proof}
	Since all nonzero multiples of a given codeword share the same distribution of $\mathbf{0}$'s run, we have $L(k,q,b)\leq \frac{q^k-1}{q-1}$. By emulating the proof of Theorem \ref{additivecase}, linear codes can be constructed with precisely $\frac{q^k-1}{q-1}$ nonzero $b$-symbol weights. Hence the result holds.
	\end{proof}
	Let $L(n,k,q,b)$ represent the maximum number of nonzero $b$-symbol weights of a code $\C$ with parameters $[n,k]$ over $\F_q$. From Theorem \ref{thm8}, it's evident that when the length $n$ is added, we obtain $L(n,k,q,b)=L(k,q,b)$. Let $l_0(k,q,b)$ denote the smallest integer such that $L(l_0(k,q,b),k,q,b)=L(k,q,b)$. Then, we have
	\begin{proposition}\label{pro3}
		Let notations be as above. Then
		$l_0(k,q,b)\geq \left\lceil\frac{bq}{2}\cdot \frac{q^k-1}{q-1}\right\rceil.$
	\end{proposition}
	\begin{proof}
		The desired result follows from Theorem \ref{thm8} and \cite[Lemma 5.1]{AN}.
	\end{proof}
	
	When $n< l_0(k,q,b)$, the function $L(n,k,q,b)$ depends on the maximum $\0$'s run $\theta$ of $\C$.
	A similar combinatorial function $L(n,k,q,b,\theta)$ can be defined as follows: the maximum size of $W_b(\C)$, where $\C$ is a linear code with parameters $n,k,q$, and $\theta$.

	Recall that $\C(i)=\{\ccc|\ccc\in \C^* {\rm ~and~} \vartheta(\ccc)=i\}$, then we have the following results.
	\begin{lemma}\label{range}
		For any positive integers $n$, $b$ and $m$, we have:
		\begin{itemize}
\setlength{\itemsep}{5pt}
\setlength{\parsep}{0pt}
\setlength{\parskip}{0pt}
			\item [{\rm (1)}] $W_{b+m}(\C(b))=\{n\};$
			\item [{\rm (2)}] $W_b(\C(b))\subseteq\left\{n-\left\lfloor\frac{n}{b+1}\right\rfloor,
			n-\left\lfloor\frac{n}{b+1}\right\rfloor+1,\ldots,n-1\right\};$
			\item [{\rm (3)}] $W_b(\C(b+m))\subseteq\left\{n-(m+1)\cdot\left\lfloor\frac{n}{b+m+1}\right\rfloor,
			n-(m+1)\cdot\left\lfloor\frac{n}{b+m+1}\right\rfloor+1,\ldots,n-(m+1)\right\};$
			\item [{\rm (4)}] $1\leq|W_b(\C)|\leq n+1-b\cdot\left\lfloor\frac{n}{\theta+1}\right\rfloor.$
		\end{itemize}
	\end{lemma}
	\begin{proof}
		The first statement is trivial and omitted. For any codeword $\ccc\in C(b)$, we have $1\leq\Psi(\mathbf{0}_b)\leq \lfloor\frac{n}{b+1}\rfloor$. It follows from Lemma \ref{wbformula11} that $n-\lfloor\frac{n}{b+1}\rfloor\leq w_b(\ccc)\leq n-1.$ Hence, the second claim holds.
		A similar argument is used in the third statement.
		The last statement follows from the first three statements.
	\end{proof}
	When $b>\theta,$ $L(n,k,q,b,\theta)=L(n,k,q,b+m_0,\theta)=1$ for any positive integer $m_0$. Hence we need to focus on the case $b\leq \theta$. The subsequent results establish the monotonicity of the function $L(n,k,q,b,\theta)$ with respect to parameters $n, k,$ and $q$.
	\begin{theorem}\label{the6}
		For all positive integers $n$, $b$, $k$, $\theta$, $m$ and all prime powers $q$, we have:
		\begin{itemize}
\setlength{\itemsep}{5pt}
\setlength{\parsep}{0pt}
\setlength{\parskip}{0pt}
			\item [{\rm (1)}]$L(n,k,q,b,\theta) \leq L(n+1,k,q,b,\theta);$
			\item [{\rm (2)}]$L(n,k,q,b,\theta) \leq L(n,k+1,q,b,\theta);$
			\item [{\rm (3)}]$L(n,k,q,b,\theta) \leq L(n,k,q^m,b,\theta).$
			\item [{\rm (4)}]$L(n,1,q,b,\theta)=1;$
			\item [{\rm (5)}]$L(n,n,q,b,n-1)=n-b+1;$
			\item [{\rm (6)}]$L(n,k,q,b,\theta)\leq n+1-b\cdot\left\lfloor\frac{n}{\theta+1}\right\rfloor.$
		\end{itemize}
	\end{theorem}
	\begin{proof}
		The first four statements are obvious. If $n=k$, then $C$ has to be a universe code with parameters $[n,n,1]_q$, and $\theta=n-1$. By the definition of $b$-symbol metric, we have $L(n,n,q,b,n-1)=n-b+1.$
		The last statement follows from the fourth statement of Lemma \ref{range}.
	\end{proof}
	
	From the proof of Theorem \ref{the6}, we observe that the parameter $\theta$ relies on the dimension $k$, the length $n$, and the algebraic structure of $\C$. For instance, in the case of cyclic codes, we have $\theta=k-1$. We will delve into the specifics of cyclic codes in Section \ref{cycliccodes} and elaborate on why $\theta=k-1$ when $\C$ is cyclic.
	From the perspective of the automorphism group of linear codes, Chen and Zhang \cite{ChenZhang} provided an upper bound on the cardinality of the Hamming weight spectrum of linear codes.	
	However, their result may not hold true if the Hamming metric is replaced with the $b$-symbol metric. Under the $b$-symbol metric, more restrictive conditions for the subgroup $\mathsf{G}$ are required.

	\begin{proposition}\label{Prop28}
		Let $\C$ be a linear code of length $n$ over $\F_q$ with $l$ nonzero $b$-symbol weights and let ${\rm Aut}(\C)$ be the automorphism group of $\C$. Assume that $\mathsf{G}$ is a subgroup of ${\rm Aut}(\C)$ such that for any $\ccc\in \C$ and $A\in \mathsf{G}$, $\ccc$ and $\ccc A$ {\bf have the same $\0$'s run distribution}. If the number of orbits of $\mathsf{G}$ on $\C^*$ is equal to $s$, then $l\leq s$. Moreover, the equality holds if and only if for any two nonzero codewords $\ccc_1,\ccc_2\in \C$ with the same $b$-symbol weight, there exists an automorphism $A\in \mathsf{G}$ such that $\ccc_1A=\ccc_2$.
	\end{proposition}
	\begin{proof}
		Since $\mathsf{G}$ always preserves the $\mathbf{0}$'s run weight distribution of the codewords in $\C$ without alteration, it can be deduced from Lemma \ref{wbformula11} that codewords belonging to the same orbit in $\C$ possess the same $b$-symbol weight. The proof of the second claim aligns with \cite[Proposition II.2]{ChenZhang}, hence we omit it here.
	\end{proof}
	
	\section{Cyclic codes}\label{cycliccodes}
	\hspace{0.5cm}In this section, we assume that $\gcd(n,q)=1$ and focus on the bounds of $b$-symbol weight spectrum of cyclic codes. Firstly, we recall some knowledge related to cyclic codes.
	\begin{definition}
		A $q$-ary linear code $\C$ of length $n$ is called cyclic if $\C$ remains invariant under the cyclic shift $\tau$, defined as:
		$$\tau(c_0,c_1,\ldots,c_{n-1})=(c_{n-1},c_0,\ldots, c_{n-2})\in \C,$$
		where $(c_0,c_1,\ldots,c_{n-1})\in \C.$
	\end{definition}

	By associating a vector $(c_0,c_1,\ldots,c_{n-1})\in\F_q^n$ with $\sum_{i=0}^{n-1}c_0x^i\in \frac{\F_q[x]}{(x^n-1)},$ a linear code $\C$ of length $n$ over $\F_q$ corresponds to a subset of the residue class $\frac{\F_q[x]}{(x^n-1)}.$ The linear code $\C$ is cyclic if and only if the corresponding subset in $\frac{\F_q[x]}{(x^n-1)}$ forms an ideal of the ring $\frac{\F_q[x]}{(x^n-1)}.$
It is well-known that every ideal of $\frac{\F_q[x]}{(x^n-1)}$ is principal. To distinguish the principal ideal $(f(x))$ of $\F_q[x]$ from that ideal in $\frac{\F_q[x]}{(x^n-1)}$, we use the notation $\langle f(x)\rangle$ for the principal ideal of $\frac{\F_q[x]}{(x^n-1)}$ generated by $f(x)$.
	Let $\C=\langle g(x)\rangle$ be a cyclic code, where $g(x)$ is monic and has the least degree. Let $h(x)=\frac{x^n-1}{g(x)}$. Then $g(x)$ and $h(x)$ are referred to as the generator polynomial and check polynomial, respectively. The dual code, denoted by $\C^{\perp}$, of $\C$ has generator polynomial $h^*(x)$, which is the reciprocal of $h(x),$ normalized to be monic.

	
	When $\C$ is cyclic, every nonzero codeword in $\C$ is generated by a recursion of degree $k$, and therefore, it has at most $k-1$ consecutive zeroes. In other words, $\theta=k-1$ if $\C$ is cyclic. Consequently, we obtain the following fundamental inequalities concerning the three critical parameters $n, k,$ and $d_H(\C)$ of cyclic codes, which are independent of any information other than the cyclic structure.
	\begin{theorem}\label{the9}
		Let $\C$ be a cyclic code with parameters $[n,k,d_H(\C)]$ over $\F_q$. Then
		$d_H(\C)\geq \left\lfloor\frac{n}{k}\right\rfloor.$
	\end{theorem}
	\begin{proof}
		Assume that $n=kt+t_0$, where $t$ and $t_0$ are non-negative integers, and $0\leq t_0\leq k-1$.
		Since the maximum $\mathbf{0}$'s run length $\theta$ of a cyclic code equals $k-1$, for any nonzero codeword $\ccc\in \C$, $w_H(\ccc)$ at least $t$. The codeword with Hamming weight $t$ is of the form
		\begin{equation*}
			\ccc=(\mathbf{0}_{k-1}, a_1, \mathbf{0}_{k-1},
			a_2, \ldots, \mathbf{0}_{k-1}, a_t, \mathbf{0}_{t_0}).
		\end{equation*}
		Therefore, we obtain the desired result.
	\end{proof}
	\begin{remark}
		Reference \cite{Hend} is a thesis focusing on lower bounds for the minimum Hamming distance of cyclic codes. To our surprise, although the result of Theorem \ref{the9} is fundamental, it has not been previously provided to the best of our knowledge. Existing results on lower bounds for the minimum Hamming distance of cyclic codes $\C$ (such as the BCH Bound \cite{BCH,BCH1,Hoc}, Hartmann-Tzeng Bound \cite{HTz}, etc.) are predicated on their defining sets. The lower bound presented in Theorem \ref{the9} solely relies on the length $n$ and the dimension $k$. In other words, compared to the BCH bound and Hartmann-Tzeng bound, the lower bound offered by Theorem \ref{the9} necessitates less information about the cyclic code $\C$. However, this bound is relatively loose as a trade-off.
	\end{remark}
	The maximum length of $\mathbf{0}$'s runs in cyclic codes can also aid in characterizing the magnitude of the $b$-symbol weight spectrum of cyclic codes.

	\begin{theorem}\label{bound}
		Let $\C$ be a cyclic code with parameters $[n,k]_q$. Then
		\begin{equation*}
			1\leq|W_b(\C)|\leq n+1-b\cdot\left\lfloor\frac{n}{k}\right\rfloor.
		\end{equation*}
	\end{theorem}
	\begin{proof}
		Combining Lemma \ref{range} with the fact that the maximum length of $\mathbf{0}$'s runs in $\C$ equals $k-1$, we derive the desired result.
	\end{proof}
	When $\C$ is cyclic, the maximum length of $\mathbf{0}$'s runs, denoted as $\theta$, equals $k-1$. If $b\geq k$, then $W_b(\C)=\{n\}$. Hence, our focus lies on the case where $b<k$. It's important to note that we do not need to consider the parameter $\theta$ again when $\C$ is cyclic. We can define a combinatorial function $\Gamma(n,k,q,b)$ derived from $L(n,k,q,b,\theta)$ as:
	\begin{itemize}
		\item $\Gamma(n,k,q,b):$ the maximum size of $W_b(\C)$, where $\C$ is a cyclic code with parameters $[n,k]_q$.
	\end{itemize}
	The following proposition gives some bounds on the function $\Gamma(n,k,q,b)$.
	\begin{proposition}\label{prop13}
		For any positive $n,k,b$, and all prime power $b$, we have the following:
		\begin{itemize}
\setlength{\itemsep}{5pt}
\setlength{\parsep}{0pt}
\setlength{\parskip}{0pt}
			\item[{\rm (1)}]$\Gamma(n,n,q,b)=n-b+1;$
			\item[{\rm (2)}]$\Gamma(n,1,q,b)=1;$
			\item[{\rm (3)}]$\Gamma(2^{k_0}-1,k_0,2,b)=1$ for any prime $k_0;$
			\item[{\rm (4)}]$\Gamma(n,k,q,b)\geq k-b+1.$
		\end{itemize}
	\end{proposition}
	\begin{proof}
		The first two statements are self-evident. If $k_0$ is prime, then the cyclic code with parameters $[2^{k_0}-1,k_0]_2$ will either be the binary simplex code $S_{k_0}$ or replicated versions of $S_{{k_0^{\prime}}}$ for $k_0$ being a proper divisor of $k_0^{\prime}$. According to Remark 4 in \cite{BUG}, $S_{k_0}$ is a single $b$-symbol weight code. Hence, the third statement holds.

By the monotonicity of $\Gamma(n,k,q,b)$ with respect to $n$, we have $$\Gamma(n,k,q,b)\geq \Gamma(k,k,q,b)=k-b+1.$$ Thus, we complete the proof.
	\end{proof}
	The subsequent result demonstrates an interesting relationship between $\Gamma(n,k,q,b)$ and $L(k,q)$:
	\begin{theorem}Let $\C$ be a cyclic code with parameters $[n,k]$ over $\F_q$.
		If $\gcd(n,q-1)=1$ and there exists a codeword $\ccc$ in $\C$ such that $per(\ccc)=n$, then $$\Gamma(n,k,q,b)\leq L(k,q)+1-b\cdot\left\lceil\frac{L(k,q)}{k}\right\rceil.$$
	\end{theorem}
	\begin{proof}
		Since $per(\ccc)=n$, then $|\{\tau^i(\ccc)|0\leq i\leq n-1\}|=n$. Assume that there exists an element $\beta\in\F_q^*\setminus\{1\}$ and an integer $i$ with $1\leq i\leq n-1$ such that $\beta\cdot\ccc=\tau^{i}(\ccc)$. According to \cite[Lemma 1]{SLNS}, then $\beta$ belongs to the unique cyclic subgroup of $\F_q^*$ of order $\gcd(per(\ccc),q-1)=1$, contradicting our assumptions that $\beta\in \F_q^*\setminus\{1\}$ and $1\leq i\leq n-1$. Then we have
		$$\left|\{\alpha\cdot\tau^i(\ccc)|0\leq i\leq n-1{\rm ~and~}\alpha\in \F_q^*\}\right|=n(q-1).$$
		Since the number of nonzero codewords in $\C$ is $q^k-1$, $n\leq \frac{q^k-1}{q-1}.$ Using the third statement of Proposition \ref{prop13} and the fact that $L(k,q)=\frac{q^k-1}{q-1}$, we have
		$$\Gamma(n,k,q,b)\leq \Gamma\left(\frac{q^k-1}{q-1},k,q,b\right)\leq L(k,q)+1-b\cdot\left\lceil\frac{L(k,q)}{k}\right\rceil.$$
		This completes the proof.
	\end{proof}
	
	To accurately determine the bounds of the $b$-symbol weight spectrum of cyclic codes, it is crucial to have relevant information about cyclic codes in advance. We can derive the bounds of the $b$-symbol weight spectrum of cyclic codes from the following three perspectives:
	\begin{itemize}
\setlength{\itemsep}{5pt}
\setlength{\parsep}{0pt}
\setlength{\parskip}{0pt}
		\item Obtain the bounds of the $b$-symbol weight spectrum of cyclic codes from its period distribution;
		\item Obtain the bounds of the $b$-symbol weight spectrum of cyclic codes from its primitive idempotents;
		\item Obtain the bounds of the $b$-symbol weight spectrum from the $b$-symbol weight calculation formula of $\ccc$, where $\ccc$ is a codeword of cyclic codes.
	\end{itemize}
	We refer to these approaches as the period distribution approach, the primitive idempotent approach, and the $b$-symbol weight formula approach.
	
	\subsection{The period distribution approach}

For a polynomial $f(x)$ over $\F_q$, the  period (or  order) of $f(x)$ is the least positive integer $t$ such that $f(x)|(x^t-1)$, denoted by $per(f)=t.$
	
	For a vector $\mathbf{c}$ in $\F_q^n$, it can be extended by repetition into a bi-infinite periodic sequence $\overline{{\mathbf{c}}}=(c_1,c_2,\ldots)$. If there exists the smallest positive integer $t$ such that $c_{i+t}=c_i$ for all integers $i$, then $t$ is called the period of $\mathbf{c}$, denoted by $per(\mathbf{c})=t$.
Let $\{B_1,B_2,\ldots,B_n\}$ be the period distribution of $\C$, where $B_i$ is the number of nonzero codewords of period $i$ with $1\leq i<n$. It's worth noting that the period mentioned in this context refers to the least period.
	The following lemma, derived from \cite{SLNS}, is quite useful.
	\begin{lemma}\cite{SLNS}\label{SLNSlemma2}
		If $\C$ is an $[n,k]_q$ cyclic code, then
		\begin{equation}\label{1Inequality}
			|W_1(\C)|\leq \sum_{t|n}\frac{B_t}{lcm(t,q-1)}\leq1+\sum_{1<t|n}\frac{B_t}{lcm(t,q-1)}.
		\end{equation}
	\end{lemma}
	For convenience, let $\Delta=n+1-b\cdot\left\lfloor\frac{n}{k}\right\rfloor$. The upper bound provided by Lemma \ref{SLNSlemma2} relies on the period distribution of cyclic codes. The subsequent result is contingent on the check polynomial $h(x)$ and necessitates the condition $\gcd(n,q)=1$.

	\begin{theorem}\label{thm35}
		If $\C$ is an $[n,k]_q$ cyclic code with check polynomial $h(x)$.
		Assume that $\gcd(n,q)=1$. The factorization of $h(x)$ is
		\begin{equation*}\label{Inequality}
			h(x)=\prod_{i=1}^mh_i(x),	
		\end{equation*}
		where $h_i(x)$ are irreducible factors of $h(x)$ over $\F_q$ with $1\leq i\leq m$. Assume $\deg(h_i(x))=s_i$ and $per(h_i(x))=t_i$.
		Then we have
		\begin{equation}\label{sbt}
			|W_b(\C)|
			\left\{
			\begin{array}{lll}
				=1,& \hbox{if  $b\geq k$;}\\
				\leq\min\left\{\Delta,\sum\limits_{b<t|n}\frac{\sum\limits_{j|t}
					\mu\left(\frac{t}{j}\right)\cdot q^{\sum_{t_i|j}s_i}}{lcm(t,q-1)}\right\}, & \hbox{if  $t_i>b$ for all $i$;} \\
				\leq\min\left\{\Delta,1+\sum\limits_{b<t|n}\frac{\sum\limits_{j|t}
					\mu\left(\frac{t}{j}\right)\cdot q^{\sum_{t_i|j}s_i}}{lcm(t,q-1)}\right\}, & \hbox{otherwise;}
			\end{array}
			\right.
		\end{equation}
		where $\mu(\cdot)$ denotes the M$\ddot{o}$bius function.
	\end{theorem}
	\begin{proof}
		If $\ccc$ is a nonzero codeword of $\C$, then $\ccc$, $\alpha\ccc$, and $\tau^i(\ccc)$ have the same $\mathbf{0}$'s run distribution for any $\alpha\in\F_q^*$ and positive integer $i$. In \cite{SLNS} the authors proved the inequality (\ref{1Inequality}) by the same idea. Then we obtain directly
		\begin{equation}\label{bb}
			|W_b(\C)|\leq \sum_{t|n}\frac{B_t}{lcm(t,q-1)}.
		\end{equation}
		Let $\{\Psi(\0_i)|1\leq i\leq n\}$ be the $\0$'s run distribution of a codeword $\ccc$. If the $per(\ccc)\leq b$, then $\Psi(\0_i)=0$ for all $i\geq b$. Using Lemma \ref{wbformula11}, we have $w_b(\ccc)=n$.
		If there exists $t_i\leq b$, then Inequality (\ref{bb}) can be further reduced to
		\begin{equation*}\label{bb1}
			|W_b(\C)|\leq 1+\sum_{b<t|n}\frac{B_t}{lcm(t,q-1)}.
		\end{equation*}
		If all $t_i>b$, we have
		\begin{equation*}\label{bb1}
			|W_b(\C)|\leq \sum_{b<t|n}\frac{B_t}{lcm(t,q-1)}.
		\end{equation*}
		If $n$ is coprime to $q$, the formula of $B_t$ was given in \cite{FFS} says
		$$B_t=\sum_{j|t}\mu\left(\frac{t}{j}\right)\cdot q^{\sum_{t_i|j}s_i}~\hbox{with $1<t$.}$$
		Combining Theorem \ref{bound} and the value of $B_t$, we get the desired result.
	\end{proof}
	
	If $\C$ is an irreducible cyclic code, then a nice upper bound of $W_b(\C)$ is given in the following.
	\begin{corollary}\label{cor36}
		Assume that $\gcd(n,q)=1$. Let $\C$ be an irreducible cyclic code with parameters $[n,k]_q$ and $h(x)$ be the parity-check polynomial of $\C$. Then
		\begin{equation*}
			|W_b(\C)|
			\left\{
			\begin{array}{ll}
				=1, & \hbox{if $per(h(x))\leq b$;} \\
				\leq
				\min\left\{\Delta,\frac{q^k-1}{lcm(per(h(x)),q-1)}\right\}, & \hbox{if $per(h(x))> b$.}
			\end{array}
			\right.
		\end{equation*}
	\end{corollary}
	\begin{proof}
		The period distribution of $\C$ is
		\begin{equation*}
			\{ B_{per(h(x))}=q^k-1 \hbox{ and $B_i=0$ with $i\neq per(h(x))$}\}.
		\end{equation*}
		Therefore, we get the desired result.
	\end{proof}
	
	The period distributions of several well-known cyclic codes are either determined or easily calculated. For simplicity, we provide the bounds of the $b$-symbol weight spectrum size of some special types of cyclic codes in the Appendix. These results rely on a wealth of known information on the period distributions of cyclic codes.
	\subsection{The primitive idempotent approach}
	
	\hspace{0.5cm}An element $e$ of $\frac{\F_q[x]}{(x^n-1)}$ satisfying $e^2=e$ is called an idempotent. When $\gcd(n,q)=1$, the ring $\frac{\F_q[x]}{(x^n-1)}$ is semi-simple. Then each cyclic code in $\frac{\F_q[x]}{(x^n-1)}$ contains a unique idempotent that generates the ideal. This idempotent is called the generating idempotent of the cyclic code. The irreducible cyclic codes (or minimal cyclic codes) of length $n$ over $\F_q$ are viewed as minimal ideals of the semi-simple algebra $\frac{\F_q[x]}{(x^n-1)}$. Every cyclic code of length $n$ over $\F_q$ is a direct sum of some irreducible cyclic codes. The polynomial $x^n-1$ can be factored into a unique product of monic irreducible polynomials in $\frac{\F_q[x]}{(x^n-1)}$, i.e.,
\begin{equation}
x^n-1=\prod_{i=1}^{s}f_i(x),
\end{equation}
where $f_i(x)$ is irreducible over $\F_q$ for $1\leq i\leq s$. Set $\hat{f_i}(x)=\frac{x^n-1}{f_i(x)}$. The generating idempotent of $\langle \hat{f_i}(x)\rangle$ is denoted by $\hat{e_i}(x)$. The idempotents $\hat{e_i}(x)$ for $1\leq i\leq s$ are called the primitive idempotents of $\frac{\F_q[x]}{(x^n-1)}$. Every irreducible cyclic code of length $n$ over $\F_q$ is generated by a primitive idempotent of $\frac{\F_q[x]}{(x^n-1)}$. The irreducible cyclic codes can also be described using the trace function.
	
	Let $f(x)$ be an irreducible factor of $x^n-1$ over $\F_q$. Suppose $f(x)$ has degree $m$, and let $\gamma$ be a root of $f(x)$. Then the set
	\begin{equation}\label{equation4}
		\C(Q,{\gamma})=\left\{\left.\sum_{i=0}^{n-1}{\rm Tr}_{Q/q}(\xi\gamma^i)x^i\right|\xi\in\F_{Q}\right\}
	\end{equation}
	represents the $[n,m]_q$ irreducible cyclic code with non-zeroes $\{\gamma^{-q^i}|0\leq i<m\}$. Moreover, if we set $n=\frac{Q-1}{N}$, where $N$ is an integer dividing $Q-1$, and recall that $\alpha$ is a primitive element of $\F_Q$, letting $\theta=\alpha^N$, then the set
	\begin{equation}\label{trace2}
		\C(Q,\theta)=\{\ccc(\beta)=({\rm Tr}_{Q/q}(\beta),{\rm Tr}_{Q/q}(\beta\theta),\ldots,{\rm Tr}_{Q/q}(\beta)\theta^{n-1})|\beta\in\F_Q\}
	\end{equation}
	also represents an irreducible cyclic code with parameters $[n,m_0]_q$, where $m_0$ is the multiplicative order of $q$ modulo $n$, and $m_0$ divides $m$.

	
	\begin{definition}
		The $q$-cyclotomic coset of $s$ modulo $q^k-1$ is defined to be the set
		$$C_s=\{s,sq,\ldots,sq^{r-1}\}({\rm mod~} q^k-1),$$
		where $r$ is the smallest positive integer such that $sq^r\equiv s({\rm mod~} q^k-1)$.  Let $C_i^{(k,Q)}=\alpha^i\langle\alpha^k\rangle$ for $i\in\Z_k$, where $\langle\alpha^k\rangle$ denotes the subgroup of $\F_Q^*$ generated by $\alpha^k$. The cosets $C_i^{(k,Q)}$ are called the cyclotomic classes of order $k$ in $\F_Q$.
	\end{definition}
	The set $C_s$ partitions the ring $\Z_{q^k-1}$ of integers into disjoint sets. There is a one-to-one correspondence between the primitive idempotents of $\frac{\F_q[x]}{(x^n-1)}$ and the $q$-cyclotomic cosets modulo $n$. Assume that $\varsigma$ is a primitive $n$-th root of unity in $\F_Q$. Then the ring $\frac{\F_Q[x]}{(x^n-1)}$ has $n$ primitive idempotents given by
	$$e_i=\frac{1}{n}=\sum_{j=0}^{n-1}\varsigma^{-ij}x^j {\rm ~~for~} 0\leq i\leq n-1.$$ Moreover, the ring $\frac{\F_q[x]}{(x^n-1)}$ has $s+1$ primitive idempotents given by
	$$\varepsilon_t=\sum_{j\in C_t}e_j {\rm ~~for~} 0\leq t\leq s.$$
	The ring $\frac{\F_q[x]}{(x^n-1)}$ can be factored into a direct sum of the minimal ideals $\frac{\F_q[x]}{(x^n-1)}\varepsilon_t$ for $0\leq t\leq s$ as follows:
	$$\frac{\F_q[x]}{(x^n-1)}=\bigoplus_{i=0}^{s}\frac{\F_q[x]}{(x^n-1)}\varepsilon_i.$$ The minimal ideal $\frac{\F_q[x]}{(x^n-1)}\varepsilon_i$ has the check polynomial $\prod_{j\in C_i}(x-\varsigma^j)$ since it can be viewed as an irreducible cyclic code. For more details on primitive idempotents and cyclic codes, please refer to \cite{ChenLiu} and Chapter 3 and Chapter 4 of \cite{Huffman}.

	Considering that a cyclic code without multiple roots can always be decomposed into the direct sum of some irreducible cyclic codes, we initially focus on the case of irreducible cyclic codes. 
	The subsequent lemma is quite useful as it aids us in obtaining bounds on the $b$-symbol weight spectrum size of irreducible cyclic codes.
	
	\begin{lemma}
		For any cyclic code $\C$ over $\F_q$, the group $\left\langle \tau, \{\sigma_a|a\in\F_q^*\}\right\rangle$ acting on $\C$ never changes its $\0$'s run distribution.
	\end{lemma}
	\begin{proof}
		For any $\mathbf{c}\in \C$ and $f\in \left\langle \tau, \{\sigma_a|a\in\F_q^*\}\right\rangle$, we have $f(\mathbf{c})\in\C$, and the two codewords $\mathbf{c}$ and $f(\mathbf{c})$ have the same run distribution of $\mathbf{0}$. Therefore, the group $\left\langle \tau, \{\sigma_a|a\in\F_q^*\}\right\rangle$ acting on $\C$ preserves its $\mathbf{0}$'s run distribution.
	\end{proof}
	\begin{theorem}\label{thm40}
		Let $\C$ be an $[n,k]$ irreducible cyclic code over $\F_q$. Suppose that $\C$ is generated by $\varepsilon_t$, where the primitive idempotent $\varepsilon_t$ corresponds to the $q$-cyclotomic coset $\{i_t,i_tq,\ldots,i_tq^{k-1}\}$. Then
		\begin{equation*}
			|W_b(\C)|
			\left\{
			\begin{array}{ll}
				=1, & \hbox{if $\frac{n}{\gcd(n,i_t)}\leq b$;} \\
				\leq
				\min\left\{\Delta,\frac{\gcd(n,(q-1)i_t)(q^k-1)}{n(q-1)}\right\}, & \hbox{if $\frac{n}{\gcd(n,i_t)}> b$.}
			\end{array}
			\right.
		\end{equation*}
	\end{theorem}
	\begin{proof}
		Since the primitive idempotent $\varepsilon_t$ corresponds to the $q$-cyclotomic coset $\{i_t,i_tq,\ldots,\\i_tq^{k-1}\},$ the check polynomial $h(x)$ of $\C$ has the form
		$$h(x)=\prod_{j=0}^{k-1}\left(x-\varsigma^{i_tq^j}\right),$$
		where $\varsigma$ is a primitive $n$-th root of unity.
		Then the period of $h(x)$ equals $\frac{n}{\gcd(n,i_t)}$. If the period of $h(x)\leq b$, then for any nonzero codeword $\mathbf{c}\in\C$, there must be no $\mathbf{0}$ in $\mathbf{c}$ with a length of $\geq b$. In this case, the $b$-symbol weight of each nonzero codeword is $n$.

If the period of $h(x)> b$, from Proposition \ref{Prop28}, we know that the number of orbits of the group $\left\langle \tau, \{\sigma_a|a\in\F_q^*\}\right\rangle$ serves as a nice upper bound of $|W_b(\C)|$. By combining this with Proposition \ref{prop13}, we derive the desired result.
	\end{proof}
	In the case of general cyclic codes, we can decompose them into a series of irreducible cyclic codes. However, since Chen and Zhang \cite{ChenZhang} have already obtained the upper bound on $W_H(\C)$ using the decomposition technique, we do not introduce this technique in this paper. Interested readers can refer to \cite{ChenZhang,Huffman} for more information on this technique.

	\begin{theorem}\label{thm42}
		Let $\mathcal{R}_n=\frac{\F_q[x]}{(x^n-1)}$ and let $\C$ be a cyclic code of length $n$ over $\F_q$. Suppose that
		\begin{equation*}
			\C=\bigoplus_{i=1}^v\mathcal{R}_ne_{t_i},
		\end{equation*}
		where $0\leq t_1<t_2<\cdots<t_v\leq s$, and that the primitive idempotent $e_{t_j}$ corresponds to the $q$-cyclotomic coset $\{i_{t_j},i_{t_j}q,\cdots,i_{t_j}q^{k_{t_j}-1}\}$ for each $1\leq j\leq v$. Then we have	the following.
		\begin{itemize}
			\item If  $\frac{n}{\gcd(n,i_{t_l})}>b$ for all $l$, then we have
			\begin{eqnarray*}
				|W_b(\C)|&\leq&
				\min\Bigg\{\Delta,
				\sum_{\substack{\{j_1,j_2,\ldots,j_u\}\subset [v]\\1\leq j_1< j_2< \cdots < j_u \leq v}}
				\frac{\gcd\left(n,i_{t_{j_1}},i_{t_{j_2}},\ldots,i_{t_{j_u}}\right)\cdot\prod\limits_{l=1}^u\left(q^{k_{t_{j_l}}}-1\right)}{n(q-1)}\cdot\\
				~&~&\gcd\left(q-1,\frac{n}{\gcd\left(n,i_{t_{j_1}}\right)},\ldots, \frac{n}{\gcd\left(n,i_{t_{j_u}}\right)}\right)
				\Bigg\}
			\end{eqnarray*}
			\item Otherwise, we have
			\begin{eqnarray*}
				|W_b(\C)|&\leq&
				\min\Bigg\{\Delta,
				1+\sum_{\substack{\{j_1,j_2,\ldots,j_u\}\\ {\rm ~satisfies~ {\bf Condition~ I}}}}
				\frac{\gcd\left(n,i_{t_{j_1}},i_{t_{j_2}},\ldots,i_{t_{j_u}}\right)\cdot\prod\limits_{l=1}^u\left(q^{k_{t_{j_l}}}-1\right)}{n(q-1)}\cdot\\
				~&~&\gcd\left(q-1,\frac{n}{\gcd\left(n,i_{t_{j_1}}\right)},\ldots, \frac{n}{\gcd\left(n,i_{t_{j_u}}\right)}\right)
				\Bigg\}
			\end{eqnarray*}
		\end{itemize}
		where {\bf Condition~I} is that $\{j_1,j_2,\ldots,j_u\}\subset [v], 1\leq j_1< j_2< \cdots < j_u \leq v$ and $\frac{n}{\gcd(n,i_{t_l})}>b$.
	\end{theorem}
	\begin{proof}
		Similar to the proof of Theorem \ref{thm35}, combined with the result of Theorem \ref{thm42} regarding the number of the orbits of $\left\langle \tau, \{\sigma_a|a\in\F_q^*\}\right\rangle$ on $\C^*$, we can easily obtain the expected result. It is worth noting that Condition I is designed to eliminate the codewords that do not have $\0$ with length exceeding $b-1$. The $b$-symbol weight of these excluded codewords is $n$, so we need to add another $1$.
	\end{proof}

	\subsection{The $b$-symbol weight formula approach}
	\hspace{0.5cm}In this subsection, we introduce how to calculate the $b$-symbol weight of an irreducible cyclic code. This will significantly contribute to improving the upper bound of $w_b(\C)$, where $\C$ refers to an irreducible cyclic code $\C(Q,\theta)$ in the form of (\ref{trace2}). Before proceeding, we need to introduce two related definitions.
	\begin{definition}
		An additive character of $\F_q$ is a nonzero function $\chi$ from $\F_q$ to the set of complex numbers such that $\chi(x+y)=\chi(x)\chi(y)$ for any $(x,y)\in\F_q^2.$ For each $b\in\F_q$, the function
		$$\chi_b(c)=e^{2\pi\sqrt{-1}{\rm Tr}_{q/p}(bc)/p}, ~~{\hbox{for all $c\in\F_q$}}$$
		defines an additive character of $\F_q.$ $\chi_1$ is called the canonical additive character of $\F_q$. The Gaussian periods are defined by
		$$\eta_i^{(k,Q)}=\sum_{x\in C_{i}^{(k,Q)}}\chi_1(x), ~~~i\in\{0,1,\ldots,k-1\}.$$
	\end{definition}
	
	We adopt the definition of (\ref{trace2}) instead of (\ref{equation4}) to describe irreducible cyclic codes, as Ding and Yang \cite{DY} used the definition of (\ref{trace2}) to provide the Hamming weight calculation formula for the codewords of irreducible cyclic codes, which depends on the values of some Gaussian periods. Zhu et al. \cite{ZHW2} further provided a formula for calculating the $b$-symbol weight of the codewords of irreducible cyclic codes based on \cite{DY}. Before introducing this calculation formula, we need to provide a definition of the set $U(b,i,N_1)$.

	\begin{definition}\label{UB}
		Let $N_1=\gcd\left(\frac{Q-1}{q-1},N\right)$. Define $U(b,i,N_1)$ be the set
		\begin{equation*}\label{UBIN}
			U(b,i,N_1)=\left\{(u_1,\ldots,u_b)\left|
			\sum_{j=1}^{b}u_j\theta^{j-1}\in C_i^{(N_1,Q)}\right. {\hbox{~and~}} (u_1,\ldots,u_b)\in\F_q^b\setminus\{\mathbf{0}\}\right\},
		\end{equation*}
		where $\theta=\alpha^N.$
	\end{definition}
	
	According to the definition of $\C(Q,\theta)$, it is an irreducible cyclic code with parameters $[n, m_0]$. Zhu {\it et al.} \cite{ZHW2} provided the $b$-symbol weight formula of irreducible cyclic codes.
	
	\begin{theorem}\cite{ZHW2}\label{coro12}
		Let $1\leq b\leq m_0-1$, and let $\ccc(\beta)$ be a codeword of the irreducible cyclic code $\C(Q,\theta)$ as in (\ref{trace2}).
		Let $N_1=\gcd\left(\frac{Q-1}{q-1},N\right)$.
		If $0\neq \beta\in C_i^{(N_1,Q)}$, then the $b$-symbol weight of $\ccc(\beta)$ is
		\begin{equation*}
			w_b(\ccc(\beta))=\frac{(q^b-1)(Q-1)}{q^bN}-
			\frac{N_1}{q^bN}\sum_{j=0}^{N_1-1}\#U(b,j,N_1)\eta_{i+j ({\rm mod~}N_1)}^{(N_1,Q)}.
		\end{equation*}
	\end{theorem}
	
	At this juncture, this calculation formula relies on the values of two invariants: one is the value of the Gaussian period, and the other is the cardinality of the set $U(b,i,N_1)$. By scrutinizing the calculation formula provided in Theorem 45, we can derive the upper bound on the cardinality of the $b$-symbol weight spectrum of irreducible cyclic codes.
	
	\begin{theorem}
		Let $\C(Q,\theta)$ be an irreducible cyclic code defined as in (\ref{trace2}). Then we have
		\begin{equation*}
			|W_b(\C(Q,\theta))|
			\left\{
			\begin{array}{ll}
				=1, & \hbox{if $b\geq m_0$;} \\
				\leq
				\min\left\{\Delta,\gcd\left(\frac{Q-1}{q-1},N\right)\right\}, & \hbox{if $1\leq b\leq m_0-1$.}
			\end{array}
			\right.
		\end{equation*}
		Moreover, $\C(Q,\theta)$ is an irreducible cyclic code with at most $\gcd\left(\frac{Q-1}{q-1},N\right)$ $b$-symbol weights.
	\end{theorem}
	\begin{proof}Let $\mathbf{c}(\beta_1)$ and $\mathbf{c}(\beta_2)$ be two nonzero codewords in $\C(Q,\theta)$.
From Theorem \ref{coro12}, $w_b(\mathbf{c}(\beta_1))=w_b(\mathbf{c}(\beta_2))$ if $\beta_1$ and $\beta_2$ belong to the same cyclotomic class $C_{i}^{(N_1,Q)}$, where $N_1=\gcd\left(\frac{Q-1}{q-1},N\right)$. Therefore, the cardinality of $W_b(\C,(Q,\theta))$ must not exceed the minimum value between $\Delta$ and $\gcd\left(\frac{Q-1}{q-1},N\right)$.
	\end{proof}
	For reducible cyclic codes, calculating the Hamming weight of their nonzero codewords becomes quite complex. It's evident that computing the $b$-symbol weight of nonzero codewords in reducible cyclic codes will be even more challenging. The difficulty of this task can be understood from \cite{ZHW3}. As is well-known, Kasami codes represent a class of reducible cyclic codes with the simplest structure. Zhu and Shi \cite{ZHW3} studied the $b$-symbol weight distribution of Kasami codes. It's apparent that even for the simplest reducible cyclic codes, extremely complex calculation formulas are required for their nonzero codewords. We encourage interested readers to explore the determination of the $b$-symbol weight distribution of reducible cyclic codes.
	\subsection{Performance analysis of three approaches}
	\hspace{0.5cm}The above three approaches have the following characteristics:
	\begin{itemize}
\setlength{\itemsep}{3pt}
\setlength{\parsep}{0pt}
\setlength{\parskip}{0pt}
		\item  The period distributions approach: (1) There are existing results on the period distribution of cyclic codes that can be leveraged to reduce computational complexity; (2) Suitable for reducible cyclic codes; (3) The results obtained by the period distributions approach and the primitive idempotent approach are equivalent.
		\item The primitive idempotent approach: The proposed upper bound can be applied to reducible cyclic codes.
		\item  The $b$-symbol weight formula approach: It is more accurate than  the period distributions approach and the primitive idempotent approach, but it is difficult to apply to the case of reducible cyclic codes.
	\end{itemize}
	
	\begin{proposition}
		The upper bounds provided by the period distributions approach and the primitive idempotent approach are identical.
	\end{proposition}
	\begin{proof}
		Both approaches rely on determining the number of orbits of the group $\left\langle \tau, \{\sigma_a|a\in\F_q^*\}\right\rangle$ on $\C^*$. Therefore, the upper bounds obtained by the two approaches are equivalent, although they may be expressed differently due to the use of different mathematical symbols.
		
		Considering that reducible cyclic codes can always be decomposed into a series of irreducible cyclic codes, we only need to verify that the results given by Corollary \ref{cor36} and Theorem \ref{thm40} are the same. This is equivalent to proving that the equation $$\frac{q^k-1}{{\rm lcm}({\rm per}(h(x)),q-1)}=\frac{\gcd(n,(q-1)i_t)(q^k-1)}{n(q-1)}$$ holds. Since ${\rm per}(h(x))=\frac{n}{\gcd(n,i_t)}$, we have
		$$\frac{q^k-1}{{\rm lcm}({\rm per}(h(x)),q-1)}=\frac{q^k-1}{{\rm lcm}\left(\frac{n}{\gcd(n,i_t)},q-1\right)}=
		\frac{q^k-1}{n(q-1)\cdot \frac{1}{\gcd(n,(q-1)i_t)}}.$$
		This completes the proof.
	\end{proof}

	The reasons for the significant difference on the result forms between Theorem \ref{thm35} and Theorem \ref{thm42} are twofold:
	\begin{itemize}
		\item[(i)] The parameters selection for the two approaches are different. The period distributions approach needs to determine the period of each $h_i(x)$; while the primitive idempotent approach selects the corresponding $i_{t_j}$. At this point, the correspondence between ${\rm per}(h(x))$ and $i_{t_j}$ is ${\rm per}(h(x))=\frac{q^{k_{t_j}}-1}{\gcd(q^{k_{t_j}}-1,i_{t_j})}$.
		
		\item[(ii)] The second reason for the differences is the different selection of auxiliary functions. in Theorem \ref{thm35}, we can see that the M$\ddot{o}$bius function $\mu(x)$ was selected as an auxiliary function in its results; In Theorem \ref{thm42}, we can find that the Euler's totient function $\phi(x)$ was selected as an auxiliary function in its result (we can see the Euler's totient function as an auxiliary function in the mathematical derivation of \cite{ChenZhang}). The Euler's totient function and the M$\ddot{o}$bius function have the following relationship:
		$$\phi(n)=\sum_{i|n}\mu(i)\frac{n}{i}.$$
	\end{itemize}
	
	We note that Chen {\it et al.} \cite{ChenFu} further improved the upper bound given in \cite{ChenZhang} by expanding subgroup $\left\langle \tau, \{\sigma_a|a\in\F_q^*\}\right\rangle$ to $\left\langle \tau, \{\sigma_a|a\in\F_q^*\},\mu_q\right\rangle$, where $\mu_a$ is a ring automorphism of $\frac{\F_q[x]}{(x^n-1)}$, which is defined as follows:
	\begin{eqnarray*}
		\mu_q: \frac{\F_q[x]}{(x^n-1)}&\longrightarrow&\frac{\F_q[x]}{(x^n-1)}\\
		\sum_{i=0}^{n-1}a_ix^i&\longmapsto&\mu_q\left(\sum_{i=0}^{n-1}a_ix^i\right)=\sum_{i=0}^{n-1}a_ix^{qi}({\rm mod}~x^n-1).
	\end{eqnarray*}
	However, the subgroup $\left\langle \tau, \{\sigma_a|a\in\F_q^*\},\mu_q\right\rangle$ proposed in \cite{ChenFu} is not suitable for studying the $b$-symbol weight spectrum. This is because for a nonzero vector $\ccc$, $\mu_q$ may alter the $\0$'s run distribution of $\ccc$.

	\section{Conclusion and future works}
	\hspace{0.5cm}In this paper, our contributions are as follows:
	\begin{itemize}
  \setlength{\itemsep}{3pt}
\setlength{\parsep}{0pt}
\setlength{\parskip}{0pt}
		\item[(1)]We studied the $b$-symbol distance spectrum of the unrestricted codes and established an interesting connection with the run of zero in periodic sequences. When $C$ is unrestricted, we also proved that $\D_b(C)$ is closely related to the difference sets with Singer parameters and Golomb ruler (Theorem \ref{the20} and Proposition \ref{pro21}).
		\item[(2)]When the research object is cyclic codes, we propose three approaches to give the upper bound of the $b$-symbol weight spectrum for cyclic codes.  When $b=1$, the upper bound of the Hamming weight spectrum for irreducible cyclic codes given by the $b$-symbol weight formula approach will be better than the upper bound given by \cite{ChenZhang,ChenFu}.
		\item[(3)]The upper bound of the $b$-symbol weight spectrum for many classical cyclic codes, such as RS codes, BCH codes, etc., is given using the period distributions approach, and these results are presented in the appendix.
		\item[(4)]As two by-products of this paper, we answer how many symplectic weights can a linear code have at most (Corollary \ref{addbound}), and provide a basic inequality among the parameters $[n,k,d_H(\C)]_q$ of cyclic codes (Theorem \ref{the9}).
	\end{itemize}
	The future works of this topic will focus on the following points:
	\begin{itemize}
  \setlength{\itemsep}{3pt}
\setlength{\parsep}{0pt}
\setlength{\parskip}{0pt}
		\item Research on the construction of M$b$SW codes and the shortest length of M$b$SW codes.
		\item Try to give the $b$-symbol weight calculation formula for general cyclic codes. We may be able to derive inspiration from \cite{DingLiuMaZeng,FengMomi,Kasami,Machangli,WangTangQiYang,Xiong1FFA,Xiong2DCC,Xiong3FFA,ZHW} to provide the $b$-symbol weight calculation formula for reducible cyclic codes.
		\item Consider the size of the $b$-symbol distance spectrum or $b$-symbol weight spectrum of codes $C$ over $\Z_4$.
	\end{itemize}
	
	\section*{Appendix}
	\hspace{0.5cm}We further give some upper bounds on the size of $W_b(\C)$ of some cyclic codes by using these known (or easily computed) period distributions.
	
	\subsubsection*{The cyclic codes with prime length $n$}
	If $\C$ is a cyclic code with prime length $n$, then a nice upper bound of $W_b(\C)$ is given in the following.
	\begin{corollary}\label{cor48}
		Assume that $n$ is prime. Let $\C$ be a cyclic code with parameters $[n,k]_q$ and $h(x)$ be the parity-check polynomial of $\C$. Assume that $per(h(x))>b$. Then we have
		\begin{equation*}
			|W_b(\C)|\leq\left\{
			\begin{array}{ll}
				\Delta, & \hbox{$n|q-1$;} \\
				\min\left\{\Delta,1+\frac{q^k-q}{n(q-1)}\right\}, & \hbox{$x-1|h(x)$ and $n\not|~q-1$;}\\
				\min\left\{\Delta,\frac{q^k-1}{n(q-1)}\right\}, & \hbox{$x-1\not|~h(x)$ and $n\not|~q-1$.}
			\end{array}
			\right.
		\end{equation*}
	\end{corollary}
	\begin{proof}
		When $n$ is prime, the period distribution of $\C$ is
		\begin{equation*}
			\left\{
			\begin{array}{ll}
				B_1=q-1, B_i=0, 2\leq i\leq n-1, B_n=q^k-q & \hbox{if $x-1|h(x)$;} \\
				B_i=0, 1\leq i\leq n-1, B_n=q^k-1 & \hbox{if $x-1\not|~ h(x)$.}
			\end{array}
			\right.
		\end{equation*}
		If $n|q-1$, we have $|W_b(\C)|\leq\min\{\Delta, \frac{q^k-1}{q-1}\}=\Delta$ since $\Gamma(k,q,b)\leq L(k,q,b)=\frac{q^k-1}{q-1}.$
		According to Inequality (\ref{bb}), we obtain the desired result.
	\end{proof}
	
	\subsubsection*{The dual of cyclic codes with two zeroes}
	Recall that $h(x)$ denotes the parity-check polynomial of $\C$.
	If $h(x)$ has only two irreducible factors, the period distribution of $\C$ is easy to compute after a brief discussion. So we omit the proof of the following result.
	\begin{corollary}
		Assume $\gcd(n,q)=1$. Let $\C$ be a cyclic code with parameters $[n,k]_q$. Let $h(x)=h_1(x)h_2(x)$ where $h_1(x)$ and $h_2(x)$ are irreducible over $\F_q$, $\deg(h_i(x))=s_i$ and $per(h_i(x))=t_i$ with $1\leq i\leq2.$
		If  $\min\{t_1,t_2\}>b$, then
		\begin{equation*}
			|W_b(\C)|\leq
			\left\{
			\begin{array}{ll}
				\min\left\{\Delta,\frac{q^k-1}{lcm(t_1,q-1)}\right\}, & \hbox{if $t_1=t_2$;} \\
				\min\left\{\Delta,\frac{q^{s_1}-1}{lcm(t_1,q-1)}+
				\frac{q^{s_2}-1}{lcm(t_2,q-1)}+
				\frac{q^k-q^{s_1}-q^{s_2}+1}{lcm(lcm(t_1,t_2),q-1)}\right\}, & \hbox{if
					$\gcd(t_1,t_2)\in\{t_1,t_2\}$;} \\
				\min\left\{\Delta,\frac{q^{s_1}-1}{lcm(t_1,q-1)}+
				\frac{q^k-q^{s_1}}{lcm(t_2,q-1)}\right\}, & \hbox{if $t_1|t_2$.}
			\end{array}
			\right.
		\end{equation*}
	\end{corollary}
	\subsubsection*{BCH codes}

	Let $\delta$ be an integer with $2\leq \delta\leq n.$ A BCH code $\C$ over $\F_q$ of length $n$ and designed distance $\delta$ is a cyclic code with defining set $$T=C_b\cup C_{b+1}\cup\cdots\cup C_{b+\delta-2},$$ where $C_i$ denotes the $q$-cyclotomic coset modulo $n$ containing $i$. The period distribution of BCH codes is given in \cite{LH} in some cases. Combining the results in \cite{LH}, we get the following result directly.
	\begin{corollary}
		Let $\C$ be a narrow-sense BCH code with length $n=q^m-1$ and designed distance $\delta$.
		We have the following.
		\begin{equation*}
			|W_b(\C)|\leq \min\left\{\Delta,
			\sum\limits_{b<t|n}\frac{\sum\limits_{j|t}
				\mu\left(\frac{t}{j}\right)A_i}{lcm(t,q-1)}
			\right\},
		\end{equation*}
		where the value of $A_i$ is as follows:
		\begin{itemize}
			\item[{\rm (1)}] If $2\leq \delta <q+1$, then
			\begin{equation*}
				A_i=\left\{
				\begin{array}{ll}
					q^{n-(\delta-1)m}, & \hbox{$i=n$;} \\
					q^{\frac{n}{2}-\lfloor\frac{\delta-1}{2}\rfloor m}, & \hbox{$i=\frac{n}{2}$;} \\
					
					A_{\frac{sn}{j}}=q^{\frac{n}{j}-\lfloor\frac{\delta-1}{j}\rfloor m}, & \hbox{$i=\frac{n}{j}$ for some $j\in [\delta-1]\backslash\{1,2\}$}\\ ~&\hbox{where $1\leq s\leq j-1$ and $\gcd(s,j)=1$;} \\
					q^{\gcd(n,i)}, & \hbox{otherwise.}
				\end{array}
				\right.
			\end{equation*}
			\item[{\rm (2)}] If
			$
			q+1\leq \delta\leq
			\left\{
			\begin{array}{ll}
				q^{\frac{m+1}{2}},  &{\rm{if}~~} m\equiv 1({\rm mod~~2}); \\
				q^{\frac{m}{2}}+2,  &{\rm{if}~~} m\equiv 0({\rm mod~~2}),\\
			\end{array}\right.
			$
			then
			\begin{equation*}
				A_i=\left\{
				\begin{array}{ll}
					q^{n-((\delta-1)-\lfloor\frac{\delta-1}{q}\rfloor)m}, & \hbox{if $i=n$;} \\
					
					A_{\frac{sn}{j}}=q^{\frac{n}{j}}-
					\left(\left\lfloor\frac{\delta-1}{j}\right\rfloor-\left\lfloor\frac{\delta-1}{qj}\right\rfloor\right)m, & \hbox{$i=\frac{n}{j}$ for some $j\in [\delta-1]$,}\\
					~&\hbox{where $1\leq s\leq j-1$ and $\gcd(s,j)=1$;} \\
					q^{\gcd(n,i)}, & \hbox{otherwise.}
				\end{array}
				\right.
			\end{equation*}
			\item[{\rm (3)}] If $q^{\frac{m}{2}}+2\leq\delta\leq2q^{\frac{m}{2}},$ then
			\begin{equation*}
				A_i=\left\{
				\begin{array}{ll}
					q^{n-\left(\delta-\frac{3}{2}-\left\lfloor\frac{\delta-1}{q}\right\rfloor\right)m}, & \hbox{$i=n$;} \\
					A_{\frac{sn}{j}}=q^{\frac{j}{n}-
						\left(\left\lfloor\frac{\delta-1}{j}\right\rfloor-\left\lfloor\frac{\delta-1}{qj}\right\rfloor\right)m}, &
					\hbox{$i=n$, $j\in [\delta-1]\backslash\{1\}$, $\gcd(s,j)=1$,} \\
					~&1\leq s\leq j-1, \left(q^{\frac{r}{2}}+1\right)\not|j;\\
					A_{\frac{sn}{j^{\prime}}}=q^{\frac{n}{j^{\prime}}-\frac{m}{2}}, & \hbox{$r=\frac{n}{j^{\prime}}, j^{\prime}\in [\delta-1], \gcd(s,j^{\prime})=1$,}\\
					~&{1\leq s\leq j^{\prime}-1, \left(q^{\frac{m}{2}}+1\right)|j^{\prime};} \\
					q^{\gcd(n,i)}, & \hbox{otherwise.}
				\end{array}
				\right.
			\end{equation*}
		\end{itemize}
	\end{corollary}
	\subsubsection*{Reed-Solomon codes}
	A Reed-Solomon code, abbreviated RS code, is a BCH code of length $n=q-1$. By virtue of the period distribution of the RS code, we obtain the following result.
	\begin{corollary}\label{cor51}
		Let $\C$ be a Reed-Solomon code with parameters $[q-1,k]_q$. Then
		\begin{equation*}
			|W_b(\C)|\leq \min\left\{\Delta,1+\sum_{b<t|q-1}\frac{\sum\limits_{j|t}\mu\left(
				\frac{t}{j}\right)q^{1+
					\frac{(k-1)\gcd(q-1,j)}{q-1}}}{q-1}\right\}.
		\end{equation*}
	\end{corollary}
	\begin{proof}
		The period distribution of Reed-Solomon code is
		\begin{equation*}
			B_t=\sum_{j|t}\mu\left(\frac{t}{j}\right)q^{1+\frac{(k-1)\gcd(q-1,j)}{q-1}},
		\end{equation*}
		which is given in \cite{FFS},
		and $lcm(t,q-1)=q-1$, then we obtain the desired result.
	\end{proof}
	From Corollary \ref{cor48} to Corollary \ref{cor51}, we can see that the period distributions approach has significant advantages over the primitive idempotent approach, thanks to the existing results on the period distribution of cyclic codes.
	\section*{Acknowledgement}
	This research is supported by Natural Science Foundation of China (12071001, 62171248).

\end{document}